\documentclass[english,aps,prb,amsmath,amssymb,twocolumn,letterpaper,showpacs,superscriptaddress,longbibliography]{revtex4-1}

\usepackage{mathtools} 

\usepackage[normalem]{ulem} 
\usepackage{comment}
\usepackage{physics}
\usepackage{graphicx,bm}
\usepackage{makecell,multirow}
\usepackage{babel}
\usepackage{esint}
\usepackage{graphicx}
\usepackage{dcolumn}
\usepackage{bm}
\usepackage{hyperref}
\usepackage[mathlines]{lineno}
\usepackage{dsfont}
\usepackage{lineno}
\usepackage{cancel}
\usepackage{xcolor}
\usepackage{booktabs}
\usepackage{longtable}
\usepackage{braket}
\usepackage{amsmath,mleftright}
\usepackage{latexsym}
\usepackage{amssymb, amsthm}
\usepackage{amstext}
\usepackage{mathrsfs}
\usepackage{amsfonts}
\hypersetup{
    colorlinks=true,
    linkcolor=red,
    citecolor=blue,
    filecolor=magenta,      
    urlcolor=blue,
    pdftitle={Overleaf Example},
    pdfpagemode=FullScreen,
    }

\newtheorem{definition}{Definition}[section]
\newtheorem{proposition}{Proposition}[section]
\newtheorem{corollary}{Corollary}[section]

\newtheorem{theorem}{Proposition}[section]
\newtheorem{lemma}{Lemma}[section]

\definecolor{JNcolor}{rgb}{0.9,0.1,0.1}

\definecolor{SSScolor}{rgb}{0,0,1}

\begin{document}

\title{Topological entanglement entropy meets holographic entropy inequalities}
\author{Joydeep Naskar}
\email{naskar.j@northeastern.edu}
\affiliation{Department of Physics, Northeastern University, Boston, MA 02115, USA}
\affiliation{The NSF AI Institute for Artificial Intelligence and Fundamental Interactions, Cambridge, MA, U.S.A.}

\author{Sai Satyam Samal}
\email{samals@purdue.edu}
\affiliation{Department of Physics and Astronomy, Purdue University, West Lafayette, IN 47907, USA}

\date{\today}

\begin{abstract}
Topological entanglement entropy (TEE) is an efficient way to detect topological order in the ground state of gapped Hamiltonians. The seminal work of Kitaev and Preskill~\cite{preskill-kitaev-tee} and simultaneously by Levin and Wen~\cite{levin-wen-tee} proposed separate definitions of TEE based on a subtraction scheme. In the present work, we explain why the subtraction schemes work for the computation of TEE and generalize them for arbitrary number of subregions by explicitly noting the necessary conditions for an information quantity to capture TEE. Our analysis puts the two definitions~\cite{levin-wen-tee,preskill-kitaev-tee} into separate classes. Focusing on cyclic information quantities $\mathcal{Q}_{2n+1}$ and multi-information $\mathcal{I}_n$, we propose a generalized framework for defining TEE. We also show that the holographic entropy inequalities are satisfied by the quantum entanglement entropy of the non-degenerate ground state of a topologically ordered two-dimensional medium with a mass gap.

\end{abstract}

\maketitle





\section{Introduction}

Understanding topological phase transitions and detecting the topological order have been one of the most exciting research directions since the discovery of the fractional quantum hall (FQH) effect~\cite{PhysRevLett.48.1559}. A topological phase transition falls beyond the Landau-Ginzburg paradigm~\cite{10.1093/acprof:oso/9780199227259.001.0001} and hence requires a new approach~\cite{10.21468/SciPostPhys.13.5.114,annurev:/content/journals/10.1146/annurev-conmatphys-040721-021029}. 
A topological phase is protected by a finite energy gap at zero temperature~\cite{10.1093/acprof:oso/9780199227259.001.0001} and closing the gap by continuously tuning the parameters in the Hamiltonian would lead to a different topological phase~\cite{Alicea_2012}. These exotic phases are associated with the presence of gapped excitations called anyons~\cite{leinaas_theory_1977,PhysRevLett.49.957,PhysRevLett.53.722,doi:10.1142/5752}. These quasi-particles have braiding statistics that is intermediate between bosons and fermions~\cite{KITAEV20062,RevModPhys.80.1083,rao2016introduction} and have recently been experimentally verified in FQH liquids~\cite{Nakamura_2020,doi:10.1126/science.aaz5601,PhysRevX.13.011031}.

However, we must point out that looking at anyonic excitations is not the only way to characterize the topological phase. Ground state degeneracy also provides another quantum number that characterizes different topological phases. It was shown that the ground state degeneracy of FQH states (filling factor $\nu=\frac{1}{m}$) goes as $m^{g}$ on a Reimann manifold with $g$ genus~\cite{PhysRevB.41.9377} and is non-degenerate on a sphere~\cite{PhysRevLett.51.605}. Interestingly, the Von-Neumann (VN) entanglement entropy of the ground state also gives information about the topological phase~\cite{preskill-kitaev-tee,levin-wen-tee}. The VN entropy of gapped ground state in the thermodynamic limit goes as the area of the system~\cite{srednicki-area-law,area-law-review} with correction, $\gamma$, that is universal and independent of the area of the sub-system,
\begin{align}\label{eq:Stopo_def1}
        S(\rho)=\alpha L - \gamma + \cdots \,.
\end{align}
The ellipsis terms vanish as the system size goes to infinity. The universal constant, $\gamma$, captures the non-local feature of the ground state and encodes information about the topological phase~\cite{preskill-kitaev-tee}.

Intuitively, in order to probe the topological entanglement entropy (TEE), one can use subtraction scheme similar to Ref.~\cite{preskill-kitaev-tee,levin-wen-tee} which is invariant under any smooth deformations. It cancels out all the boundary contributions in Eq.~\eqref{eq:Stopo_def1} and leaves only the non-local contributions (proportional to $\gamma$) to the TEE. This motivates us to ask the following question: are there other information quantities that can be regarded as equivalent definitions of the TEE and thereby detecting topological order and if yes, what are the general properties of such information quantities? By borrowing ideas from holography, we propose new information quantities that can be chosen as equivalent definitions of the TEE and also comment on the necessary conditions which must be satisfied for being a topological information quantity. We will be restricting ourselves to case where the ground state is non-degenerate or equivalently our system is on a disk in two spatial dimensions. Since the ground state of the gapped Hamiltonian can be described by an effective topological quantum field theory (TQFT). We will be presenting a TQFT calculation on a disk-like geometry following Ref.~\cite{preskill-kitaev-tee} and show that these are new topological information quantities.

\subsection{Definition of topological entanglement entropy}
In the literature, there is no consensus on the \emph{precise} definition of TEE. We will review the three commonly used definitions as follows:
\begin{itemize}
    \item The first definition comes from the area law, Eq.~\eqref{eq:Stopo_def1}. The additive constant term $-\gamma<0$ is defined as the TEE, and is believed to capture some global features of entanglement in the ground state. We will use this quantity $\gamma$ as the definition of TEE, i.e, $S_{\text{topo}}=\gamma$. Our sense of \emph{`topological'} means that $S_{\text{topo}}$ depends only on global features and not the system's dimensions, when talking about area law. In the context of anyons, we mean that $S_{\text{topo}}$ only depends on the total quantum dimension $\mathcal{D}$ and $\gamma=\log{\mathcal{D}}$. It is in some sense a strong assumption that the ground state is described by a TQFT. Nevertheless, we will proceed with this.
    
    \item The second definition comes from the work of Levin and Wen~\cite{levin-wen-tee} where TEE is defined as a tripartite information quantity called the conditional mutual information, Eq.~\eqref{eq:Stopo_def2},
    \begin{equation}\label{eq:Stopo_def2}
    \begin{aligned}
        S_{\text{topo}}^{\text{LW}}=&S(AB)+S(BC)-S(B)-S(ABC)\,,
    \end{aligned}
    \end{equation}
     computed on a geometry as shown in Fig.~\ref{fig:kitaev_wen}b. In the Eq.~\eqref{eq:Stopo_def2}, $S(X)$ is the VN entropy associated with region $X$. However, on changing this particular geometry (for example, permuting labels of regions or choosing a disk) renders this quantity useless. We would like to point out that the information quantity used in this case i.e., conditional mutual information is related to the entropy inequality called the \emph{strong sub-additivity} (SSA)~\cite{ssa1973}.

    \item The work of Kitaev and Preskill~\cite{preskill-kitaev-tee} motivated a third definition of the TEE involving a different tripartite information quantity related to an entropy inequality called \emph{monogamy of mutual information} (MMI)~\cite{MMI2013} as follows,
    \begin{equation}\label{eq:Stopo_def3}
    \begin{aligned}
        S_{\text{topo}}^{\text{KP}}=&S(AB)+S(AC)+S(BC)\\
        &-S(A)-S(B)-S(C)-S(ABC),
    \end{aligned}
    \end{equation}
    where we choose the sign convention to set $S_{\text{topo}}\geq 0$. This was first evaluated on a disk (see Fig.~\ref{fig:kitaev_wen}a and a non-simply connected geometry, Fig.~\ref{fig:kitaev_wen}b. It was found to be invariant under permutation of labels of regions. The topological property of $S_{\text{topo}}^{\text{KP}}$ holds for different geometries, however, the value $S_{\text{topo}}^{\text{KP}}$ depends on the topology (more specifically, on the number of connected components of the geometry). We call such quantities \emph{fixed-topology TEE probes}. The Ref.~\cite{Bao:2015gapped} gives an explicit formula of $S_{\text{topo}}^{\text{KP}}$ in terms of the zeroth Betti number. 
\begin{figure}
    \centering
    \includegraphics[width=1.0\columnwidth]{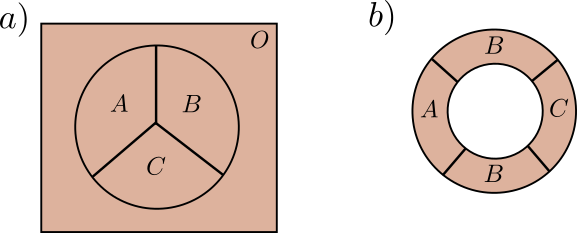}
    \caption{(a) Geometry used by Kitaev and Preskill~\cite{preskill-kitaev-tee} for computing the TEE. The information quantity used here comes from the entropy inequality called the monogamy of mutual information (MMI)~\cite{MMI2013}. (b) Geometry used by Levin and Wen~\cite{levin-wen-tee} for obtaining the TEE. In this case the information quantity used is known as conditional mutual information and it is related to entropy inequality known as strong sub-additvity (SSA)~\cite{ssa1973}. }
    \label{fig:kitaev_wen}
\end{figure}
\end{itemize}

\subsection{Summary of results}
A holography-inspired approach~\cite{Bao:2015gapped} showed that the cyclic family of holographic entropy inequalities (HEI)~\cite{Bao:2015HEC} gives the TEE, $\gamma$. This was demonstrated by explicitly applying the area law, Eq.~\eqref{eq:Stopo_def1}. In this current work, we extend the previously done analysis~\cite{Bao:2015gapped,Bao:2015HEC} in the following sense,
\begin{itemize}
    \item We don't explicitly use the area law, Eq.~\eqref{eq:Stopo_def1} and instead perform a TQFT calculation on a 2d disk-like geometry. We demonstrate that the TEE depends only on the total quantum dimension $\mathcal{D}$, in agreement with the results of Ref.~\cite{preskill-kitaev-tee,levin-wen-tee,Bao:2015gapped}.
    
    \item We show that \emph{all} HEI~\cite{Bao:2024all-hei} are satisfied by gapped Hamiltonians with unique ground state, Eq.~\eqref{eq:Stopo_def1} and, especially the \emph{facet} inequalities, can be used to evaluate the TEE, $\gamma$.
    
    \item We further show that a certain class of entanglement entropy inequalities that need not hold true for holographic quantum states, could hold true for the 2d-disk geometry~\cite{preskill-kitaev-tee}. Hence, these entropy relations can also be used to evaluate TEE.
    
    \item We contrast the definitions of TEE in Eq.~\eqref{eq:Stopo_def2} and Eq.~\eqref{eq:Stopo_def3} and classify all possible definitions of $S_{\text{topo}}$ into \emph{two} classes, suggesting necessary conditions on the \emph{subtraction scheme}~\cite{Kim:2023-tee-bound,Levin:2024proof2} of entanglement entropies.
\end{itemize}

Many of our results are straightforward consequences of the following proposition combined together with established results from the holographic entropy cone literature,
\begin{proposition}
    \label{con-2bal-satisfy}
    Any $2$-balanced information quantity $\mathcal{Q}$ is topological with indefinite sign.
\end{proposition}

We give a heuristic proof of this proposition in \ref{app:heuristic-proof} and anyonic derivation of the same in \ref{app:anyonic-topo-sb}. We give various evidence in favour of the validity of our results independent of this proposition. The organization of this paper is as follows: In section, Sec.~\ref{sec:tqft-tech}, we begin by reviewing the necessary TQFT tools for our discussion. Our main results are in section, Sec.~\ref{sec:hei-tee}, where we consider the HEIs in the context of TEE and use these results in subsection~\ref{subsec:tee-probe-tp} to discuss entanglement-based probes of topological phases. Finally, we discuss our results and outline some future directions in section, Sec.~\ref{sec:discussion}. Throughout this paper, we will neglect the spurious TEE which often arise in the symmetry protected topological phases~\cite{PhysRevB.94.075151,PhysRevB.98.075131,PhysRevLett.122.140506,PhysRevB.100.115112,PhysRevResearch.2.032005,PhysRevLett.131.166601}. However, we will bring them into the discussion towards the end in the section, Sec.~\ref{sec:discussion}.

\section{Necessary TQFT techniques}
\label{sec:tqft-tech}
In this section, we review the calculation for TEE ($S_{\text{topo}}$) for a 2d disk-like geometry with $3$ subregions, Fig.~\ref{fig:kitaev_wen}a~\cite{preskill-kitaev-tee}. The procedure involves stitching together a time-reversal (TR) copy of the system with itself and making punctures at each intersection point~\cite{preskill-kitaev-tee}. Each puncture can host anyon of a particular type. As a result of making punctures, the VN entropy of the subregion $A$ ($S_{A}$), (see Fig.~\ref{fig:kitaev_wen}a) is proportional to the VN entropy of the sphere with $3$ punctures ($\mathcal{S}_{3}$), i.e. $2S_{A}=\mathcal{S}_{3}$. The factor of $2$ comes from the fact that we have doubled the system by attaching a TR copy of the system with itself. Note that, in this procedure we don't create any anyons and therefore the net anyonic charge is zero~\cite{preskill-kitaev-tee}. Hence, the VN entropy of the sphere with $3$ punctures ($\mathcal{S}_{3}$) is equal to the VN entropy of system with $3$ anyonic excitation such that the the total charge vanishes.

Now, we proceed with the calculations for the VN entropy for a sphere with punctures. For concreteness, let us assume that we have a sphere with $4$ punctures and later we give a general formula for the sphere with $n$ punctures. Each puncture hosts anyons and let us denote the anyons by $a$, $b$, $c$ and $d$ with quantum dimension $d_{i}$ where $i=a,b,c,d$. Since, we have the constraint that the total anyonic charge in the system is zero, this implies that the VN entropy $\mathcal{S}_{4}$ is given as,
    \begin{align}\label{eq:vnentropysn}
        \mathcal{S}_{4} = -\sum_{abcd}\sum_{\mu=1}^{N_{abcd}^{\mathds{1}}} \frac{P_{abcd}^{\mathds{1}}}{N_{abcd}^{\mathds{1}}}\log\bigg(\frac{P_{abcd}^{\mathds{1}}}{N^{\mathds{1}}_{abcd}}\bigg) \,,
    \end{align}
where $P_{abcd}^{\mathds{1}}$ refers to the probability of fusing $a$, $b$, $c$ and $d$ to identity ($\mathds{1}$) and $N_{abcd}^{\mathds{1}}$ denotes the number of ways all four fuses to identity ($\mathds{1}$). Here, identity ($\mathds{1}$) refers to vacuum or an anyon with trivial anyonic charge. By using the fact that the probability of fusing anyons of type $a$ and $b$ to give $c$ is equal to $P_{ab}^{c} = \frac{N_{ab}^{c}d_{c}}{d_{a}d_{b}}$~\cite{preskilltqft}, one obtains the probability of fusing $a$, $b$, $c$ and $d$ to identity as follows,
    \begin{align}\label{eq:pacd}
        P_{abcd}^{\mathds{1}} = \frac{N^{\mathds{1}}_{abcd}d_{a}d_{b}d_{c}d_{d}}{\mathcal{D}^{6}}\,,
    \end{align}
where $\mathcal{D}=\sqrt{\sum_{i}d_{i}^{2}}$ is the total quantum dimension. Substituting Eq.~\eqref{eq:pacd} into Eq.~\eqref{eq:vnentropysn} and using the expression for the fusion probability, $P_{ab}^{c}$ one obtains the final expression for the VN entropy as,
    \begin{align}
        \mathcal{S}_{4} = 6\log D-4\sum_{a}P_{a}\log d_{a} \,,
    \end{align}
where $P_{a}=\frac{d_{a}^{2}}{\mathcal{D}^{2}}$ is interpreted as the probability of finding an anyon of type $a$ in a gas of anyons in thermal equilibrium where effectively no more fusions between any pair of anyons are taking place~\cite{preskilltqft}. One can generalize the obtained formula for VN entropy of a sphere with $4-$punctures to sphere with $n-$punctures,
    \begin{align}\label{eq:sn_entropy}
        \mathcal{S}_{n} = 2(n-1)\log \mathcal{D} - n\mathcal{K} \,,
    \end{align}
where we have defined $\mathcal{K} = \sum_{a}P_{a}\log d_{a}$ which we will use in the subsequent calculations.

Now, we can proceed with the calculation for the $S_{\text{topo}}$ in the geometry proposed in Ref.~\cite{preskill-kitaev-tee} with three regions $A$, $B$ and $C$. One can probe the TEE ($S_{\text{topo}}$) by computing the multi-information (we will give a standard definition later), $S_{\text{topo}}=-I(A:B:C)\equiv-\mathcal{I}_{3}=-2S_{3}+\frac{3}{2}S_{4}=\log \mathcal{D}$~\cite{preskill-kitaev-tee}. Hence, TEE only depends on the total quantum dimension. We would like to generalize the calculation for a disk with an arbitrary number of subregions. We begin by considering a disk with $5-$subregions and we will calculate $\mathcal{I}_{5}\equiv -I(A:B:C:D:E)$. Heuristically, one can verify $\mathcal{I}_{5}$ also successfully probes the TEE by plugging Eq.~\eqref{eq:Stopo_def1} into the formula for $\mathcal{I}_{5}$ and observing that all the boundary terms cancels exactly. On a physical ground we expect that we get the same $S_{\text{topo}}$ as we obtained in the case with $3-$subregions by computing $\mathcal{I}_{3}$. This is because the underlying topology is unchanged and increasing the number of subregions should not affect the TEE.

However, we would like to do a TQFT calculation to show that the we obtain the same TEE in the case of $5-$subregions. Let us begin by labeling the $5-$subregions as $A$, $B$, $C$, $D$ and $E$. We will now point out a subtle issue while probing the TEE using $\mathcal{I}_{5}$. The expression for $\mathcal{I}_{5}$ involves VN entropy of two types of regions, which we call contiguous (e.g. AB) and non-contiguous regions (e.g. AC). contiguous regions share a boundary while non-contiguous regions have no boundary in common. In the case where we have $3-$subregions, there is no possibility of non-contiguous regions and hence both $S_{AB}$ and $S_{AC}$ were identified with $\mathcal{S}_{4}$. With $5-$subregions, it turns out that one cannot simply identify $S_{AC}$ with $S_{5}$ and $S_{AB}$ with $S_{4}$ and carry out the calculations. The reason is, making a puncture at point where $3-$subregions meet is different from the point where $5-$subregions meet. We can circumvent the problem by smoothly deforming the geometry such that we only have intersection of at most $3-$subregions, as shown in Fig.~\ref{fig:I5}b.

\begin{figure}
    \centering
    \includegraphics[width=1.0\columnwidth]{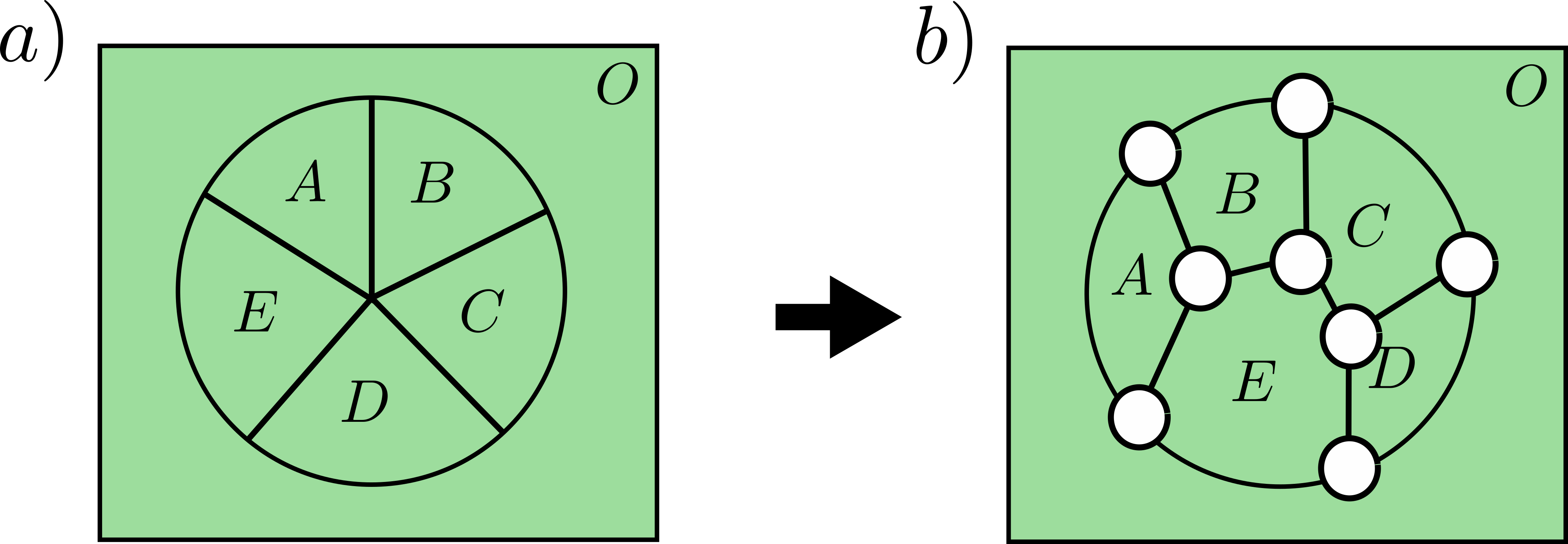}
    \caption{Geometry considered for computing the topological entanglement entropy (TEE) with $5$ subregions. By stitching a time-reversal (TR) copy of the system at spatial infinity and making punctures at the intersection points of subregions, one can map the Von Neumann (VN) entropy of a subregion to the VN entropy of sphere with $n$ number of punctures~\cite{preskill-kitaev-tee}. (a) Disk partitioned into $5$ subregions with all the subregions meeting at a given point. In order to compute the TEE, one first needs to smoothly deform it to a new geometry which have intersection of at most $3$ subregions. Note that this smooth deformation doesn't give rise to any topological phase transition and the system remains in the same topological phase. (b) The new deformed geometry with intersection of at most $3$ subregions. By introducing punctures at the intersections of the subregions after gluing with TR copy, one finds the TEE, $S_{\text{topo}}=\log \mathcal{D}$.}
    \label{fig:I5}
\end{figure}

A smooth deformation doesn't introduce any topological phase transition and system remains in the same topological phase (protected by a gap). This smooth deformation is equivalent to choosing a new set of subregions and doesn't affect the underlying Hamiltonian. By considering small deformations to the geometry, Fig.~\ref{fig:I5}b similar to Ref.~\cite{preskill-kitaev-tee} and using the Eq.~\eqref{eq:Stopo_def1} one finds that $\mathcal{I}_{5}$ probes the $S_{\text{topo}}$ successfully. This enables us to compute the TEE using the five-party multi-information, $\mathcal{I}_{5}$. A careful TQFT calculation reveals that, $S_{\text{topo}} = -\mathcal{I}_{5} = -S_{3} - S_{4} - \frac{3}{2}S_{8} + S_{7} +  2S_{6} = \log \mathcal{D}$~\cite{Note1}. This is the same as we obtained using disk with $3-$subregions and computing $\mathcal{I}_{3}$, which we expect as the topology is the same.

\section{Holographic entropy inequalities in TQFT}
\label{sec:hei-tee}
We will briefly review the idea of area law in holography to motivate the holographic entropy inequalities. The $AdS_{d+1}/CFT_{d}$ correspondence~\cite{Maldacena:1997ads-cft} is a duality between a quantum theory of gravity in $d+1$ dimensional Anti-de Sitter $(AdS)-$spacetime and a conformal field theory (CFT) on its $d$ dimensional boundary. The entanglement entropy $S(A)$ of a subregion $A$ on the boundary is related to the area of the minimal surface $\Gamma_A$ in the bulk homologous to $A$, and is given by the Ryu-Takayanagi formula~\cite{Ryu:2006bv,Ryu:2006ef},
\begin{equation}\label{eq:RT-formula}
    S(A)=\frac{|\Gamma_A|}{4G_N}\,,
\end{equation}
to the leading order, where 
$G_N$ is Newton's constant.

We will refer to the quantum states on the boundary CFT that have a semi-classical dual in the bulk and obey the RT formula, Eq.~\eqref{eq:RT-formula} as \emph{holographic states}. These states are a proper subset of all quantum states on the boundary CFT. The goal of the holographic entropy cone program ~\cite{Bao:2015HEC,Czech:2023ewn,He:2019repackaged,HernandezCuenca:2019wgh,He:2020superbalance,Hernandez-Cuenca:2023iqh,Bao:2024toric,Bao:2024contraction_map,Bao:2024all-hei} is to non-trivially constrain these multipartite holographic states. This program has led to the discovery of many holographic entropy inequalities and here we will study those inequalities in the context of TEE.

\subsection{The holographic entropy cone}
The holographic entropy cone (HEC) constrains the entropic phase space of holographic states whose leading order entanglement entropy is given by the RT formula, Eq.~\eqref{eq:RT-formula}. The HEC is a convex, rational, polyhedral cone, which can be described by a finite number of facet inequalities (see definition \ref{def:facet}). These inequalities are tight and they cannot be improved further. First, let us give an example using the five-party cyclic inequality, and then we will generalize our claims for all facet inequalities of the HEC.

Consider a five-party information quantity, $\mathcal{Q}_{5}$, which is given by Eq.~\eqref{eq:cyclic-5}. The quantity $\mathcal{Q}_{5}$ along with other information quantities formed using the cyclic family of holographic inequalities were shown to give the TEE~\cite{Bao:2015gapped} by assuming the area law, Eq.~\eqref{eq:Stopo_def1}. Here, we verify the claim by calculating $\mathcal{Q}_{5}$ explicitly using TQFT tools, Eqs.~\eqref{eq:cyclic-5tqft1}--\eqref{eq:cyclic-5tqft2}. 
\begin{widetext}
\begin{align}
\mathcal{Q}_5 &=  S(ABC) + S(BCD) + S(CDE) + S(DEA) + S(EAB) - S(AB) - S(BC) - S(CD)- S(DE)  - S(EA) \notag \\ 
& \ \ \ - S(ABCDE) \,, \label{eq:cyclic-5} \\ 
&=\frac{1}{2}\big[ 3\mathcal{S}_{7} + 2\mathcal{S}_{6} - 3\mathcal{S}_{6} - 3\mathcal{S}_{5} \big] \label{eq:cyclic-5tqft1} \,,\\
&= \frac{1}{2}\big[36\log \mathcal{D} - 21\mathcal{K} + 20\log\mathcal{D} - 12\mathcal{K}  - 30\log \mathcal{D} + 18\mathcal{K} -24\log \mathcal{D} + 15\mathcal{K}\big] =\log \mathcal{D}\,. \label{eq:cyclic-5tqft2}
\end{align}
\end{widetext}
In fact, it is straightforward to do a similar proof for the family of cyclic inequalities, see appendix~\ref{sec:cycli_inequalities}. Having illustrated an example, we will now generalize our claims about holographic entropy inequalities (HEI). We will begin by giving some definitions.

\begin{definition}[Holographic entropy inequality]
A holographic entropy inequality (HEI) is an information quantity $\mathcal{Q}$ on $n$-subregions, such that all $n$-party holographic states satisfy
\begin{equation}
    \mathcal{Q}=\sum_{i=1}^{2^n-1} a_i S_i \geq 0,
\end{equation}
where $S_i$ are the entanglement entropy of subregions accompanied by integer coefficients $a_i$.
\end{definition}

Given $n$ parties, one can construct an entropy basis of all possible ($2^n-1$) subregion entanglement entropies and order it lexicographically. For example, if $n=3$ and the subregions are $\{A,B,C\}$, then the entropy basis is $\{S_A,S_B,S_C,S_{AB},S_{AC},S_{BC},S_{ABC}\}$. An \emph{entropy vector} is a vector in this basis. We will call it a \emph{holographic entropy vector} when the system is a holographic system, i.e, the $n$ subregions are chosen on a boundary $\partial X$ of a bulk manifold $X$. The \emph{holographic entropy cone} is the space of all allowed holographic entropy vectors.

\begin{definition}[Facet HEI]\label{def:facet}
 A holographic entropy inequality $\mathcal{Q}$ involving $n$-parties is a facet of the HEC iff there exists a codimension-1 set of linearly independent holographic entropy vectors which saturate $\mathcal{Q}$. We will refer to such HEIs as `facet HEIs'.
\end{definition}

We will now define the $n$-partite information quantity $\mathcal{I}_n$ introduced already, which we will call \textit{multi-information}. Given $n$ regions $A_i$, where $i=1,\dots , n$, we define $\mathcal{I}_n$ to be,
\begin{align}\label{eq:In}
    \mathcal{I}_{n} =& \sum_{i}^n S(A_i) - \sum_{i<j}^n S(A_iA_j) + \cdots \notag\\
    &+ (-1)^{n+1} S(A_1\cdots A_n)\,,
\end{align}
which for up to $n=3$ are given by (using shorthand $S(A_i)\equiv S_{A_i}$),
\begin{equation}\label{eq:multiinfodef}
    \begin{split}
         \mathcal{I}_1=& S_{A_1} \ ; \ \mathcal{I}_2= S_{A_1}+S_{A_2}-S_{A_1 A_2}\,,\\
         \mathcal{I}_3=& S_{A_1}+S_{A_2}+S_{A_3}-S_{A_1 A_2}-S_{A_1 A_3}-S_{A_2 A_3}\\
        &+S_{A_1 A_2A_3}\,.
    \end{split}
\end{equation}
Note that $\mathcal{I}_n$ is generally sign indefinite for $n>3$ in holographic systems and the results from \cite{preskill-kitaev-tee} indicates that $\mathcal{I}_3$ is sign-definite. Here, we find that $\mathcal{I}_n$ is sign definite and equals $-\log{\mathcal{D}}$, for all $n\geq 3$. In appendix~\ref{sec:review-tee}, we explicitly calculate $\mathcal{I}_{3}$, $\mathcal{I}_{4}$ and $\mathcal{I}_{5}$. We introduce another quantity $\mathcal{I}_{m,n}$, which is defined same as $\mathcal{I}_{m}$ but for $m\leq n$, where $n$ is the total number of subregions, such that $\mathcal{I}_{n,n}\equiv\mathcal{I}_{n}$. For example, consider a system with $5$ subregions ($A_{1}$, $A_{2}$, $A_{3}$, $A_{4}$ and $A_{5}$) and arbitrarily choose $3$ subregions, say, $A_{1}$, $A_{2}$ and $A_{3}$ then $\mathcal{I}_{3,5}=\mathcal{I}_{3}$, Eq.~\eqref{eq:multiinfodef}. The set of multi-information $\{\mathcal{I}_{i,n} \}_{i=3}^{n}$ forms a basis for the facet HEI~\cite{He:2019repackaged} involving $n-$parties (excluding subaddivity).  With the above definitions, we can now formulate our following propositions and colloraries, which follow from the proposition \ref{con-2bal-satisfy}, which we will discuss later. While we present our work in a manner where our propositions are conditional on proposition \ref{con-2bal-satisfy}, we give independent supporting evidence for these propositions in the appendices. For brevity, we simply state our conditional propositions as propositions.

\begin{proposition}
\label{thm:In-topological}
All $\mathcal{I}_{m,n}$ are topological for $3\leq m \leq n$ on a $2d$ disk geometry divided into $n$ subregions.
\end{proposition}
A very simple corollary of the above proposition is the following,
\begin{corollary}\label{cor:In-topological-weak}
$\mathcal{I}_n$ is topological for $n\geq 3$ on a $2d$ disk geometry divided into $n$ subregions.
\end{corollary}
Based on the above proposition (which we prove later by using proposition~\ref{con-2bal-satisfy}), we will make a few more propositions. See the supplementary material for more discussions around the proposition~\ref{thm:In-topological} and corollaries  and~\ref{cor:In-topological-weak} and their independent verification without using proposition~\ref{con-2bal-satisfy}.
\begin{proposition}\label{thm:hei-topological}
All facet HEIs of the HEC $(n\geq 3)$ are topological, when evaluated on a disk using the quantum entanglement entropy in the unique ground state of a topologically ordered two-dimensional medium with a mass gap. 
\end{proposition}
\begin{proof}
From proposition \ref{thm:In-topological}, we have that all $\mathcal{I}_{m,n}$ are topological. All facet HEIs can be expressed in the $\mathcal{I}$-basis~\cite{He:2019repackaged}. Since the basis elements $\mathcal{I}_{m,n}$ is topological. Therefore, all facet HEIs are topological.
\end{proof}
We have proved that the facet HEIs are topological, now we will prove that it is also positive, on a 2d disk geometry. We will use the following lemma \ref{lem:coeff-hei}.
\begin{lemma}
\label{lem:coeff-hei}
    The sum of coefficients $\{a_i\}$ of all facet HEIs, except subadditivity(SA) is negative. 
\end{lemma}
This is proved in Ref.~\cite{Bao:2024all-hei}. For more details, see appendix~\ref{sec:HEI-proof}.
\begin{proposition}\label{thm:hei-satisfy}
All facet HEIs of the HEC $(n\geq 3)$ are satisfied by the quantum entanglement entropy in the unique ground state of a topologically ordered two-dimensional medium with a mass gap. 
\end{proposition}
\begin{proof}
From proposition \ref{thm:hei-topological}, we have all facet HEIs are topological. We will now prove that they are also non-negative. 
Following section~\ref{sec:tqft-tech} and appendix~\ref{sec:review-tee}, we have $\gamma=\log{\mathcal{D}}$ on the $2d$-disk geometry. 

Since we have a connected geometry, each term contributes a $-2\gamma$.
This is because each term has the form given by Eq.~\eqref{eq:sn_entropy}, and due to topological nature of HEIs shown in proposition~\ref{thm:hei-topological}, the $n$-dependence of the HEI cancels out, leaving behind $-2\gamma$, which we divide by $2$ to account for the TR copy.

By using Lemma~\ref{lem:coeff-hei} we have the sum of coefficients is negative for all facet HEIs of the HEC ($n\geq 3)$~\cite{Bao:2024all-hei}. 

Therefore, the HEIs are strictly positive in this geometry.
\end{proof}

We will take a step further and discuss the case when these inequalities are saturated. Now consider the 2d planar geometry where each region is disconnected. In this case, each $k$-party entanglement term contributes $k$ number of $-\gamma$ terms as there are $k$ boundaries. We know from \cite{He:2020superbalance} that all facet HEIs of the HEC are \emph{balanced}, i.e, each party appears equal number of times with positive and negative coefficients.  As a result, the topological terms cancel and the HEIs are strictly zero.

We have an alternative proof for the same based on proposition \ref{con-2bal-satisfy} stated after lemma \ref{lem:balanced-hei}. We give a third proof of this argument in appendix~\ref{sec:HEI-proof} exploiting the structural forms of HEIs. We have verified numerically that this holds for all known superbalanced facet HEIs.

\subsection{What about non-facet inequalities?}
We will now look at the true HEIs that are not facets of the HEC. For example, consider the following information quantity,
\begin{widetext}
    \begin{align} \label{eq:q61}
   \mathcal{Q}_{6,1} & = S(AEF)+S(BEF)+S(ADE)+S(ADF) +S(BDE)+S(BDF)+S(ABCD)+S(ABCE)\\
             & \ \ \ +S(ABCF)+S(C)-S(ABCEF)-S(ABCDF)-S(ABCDE)-S(AD)-S(AE)-S(AF)\\
             & \ \ \ -S(BD)-S(BE)-S(BF)-S(CDEF)\,, \\
             & 
             =\frac{1}{2}\big[ 4\mathcal{S}_{9} - 4\mathcal{S}_{7}+\mathcal{S}_{10}-\mathcal{S}_{6}\big] 
             = 12\log\mathcal{D} - 6\mathcal{K}\ \ \ \textbf{(Not topological)} \,.\label{eq:q61nottopological}
    \end{align}
\end{widetext}

Calculating the information quantity $\mathcal{Q}_{6,1}$ on the disk-like geometry, Fig.~\ref{fig:six_sub_regions} using Eq.~\eqref{eq:sn_entropy}, we find that $\mathcal{Q}_{6,1}$ is not topological, Eq.~\eqref{eq:q61nottopological}. 
However, it is possible to find some fixed configuration(s) of geometry where $Q_{6,1}$ is topological and gives the TEE. We classify all such information quantities that are topological only under specific choice of geometries as \emph{fixed-geometry TEE probes} (in contrast to fixed-topology TEE probes). The definition of $S_{\text{topo}}$ using \emph{strong sub-additivity} in \cite{levin-wen-tee} is another example of this class. More explicitly, we can check that on a $2d-$disk, we have
\begin{align}
    S^{\text{LW}}_{\text{topo}} &= S(AB)+S(BC)-S(B)-S(ABC) \\
    &= \frac{1}{2}\big[ 2\mathcal{S}_{4} - \mathcal{S}_{3} - \mathcal{S}_{3} \big] = 2\log \mathcal{D} - \mathcal{K},
\end{align}
whereas the same quantity evaluated on a torus-like geometry with the subregion $B$ disconnected gives the TEE (see Fig.~\ref{fig:kitaev_wen}b).

Now we will look at another example where the given inequality does not hold for all holographic states, but the information quantity defined using it, $\mathcal{Q}_{6,2}$ turns out to probe the TEE, Eq.~\eqref{eq:q62}. The factor of $2$ in front of $\log{\mathcal{D}}$ comes from the sum of coefficients being $-2$ [we will illuminate this point in~\ref{thm-general-satisfy}]
\begin{widetext}
\begin{align} \label{eq:q62}
    \mathcal{Q}_{6,2}=&S(ADE)+S(ADF)+S(AEF)+S(BDE)+S(BDF)+S(BEF)+S(CDE)+S(CDF)+S(CEF)\\
    &+S(ABC)-S(AD)-S(AE)-S(AF)-S(BD)-S(BE)-S(BF) -S(CD) -S(CE)-S(CF)\\
    & -2S(DEF)-S(ABCDEF) \,,\\
    =& \frac{1}{2}\big[4\mathcal{S}_{9} - 3\mathcal{S}_{6} - 4\mathcal{S}_{7} + \mathcal{S}_{10}\big] \,,\\
    =& \frac{1}{2}\big[4 \log\mathcal{D} \big] = 2\log\mathcal{D} \ \ \ \textbf{(Topological)}\,.
\end{align}
\end{widetext}

\begin{figure}
    \centering
    \includegraphics[width=0.4\columnwidth]{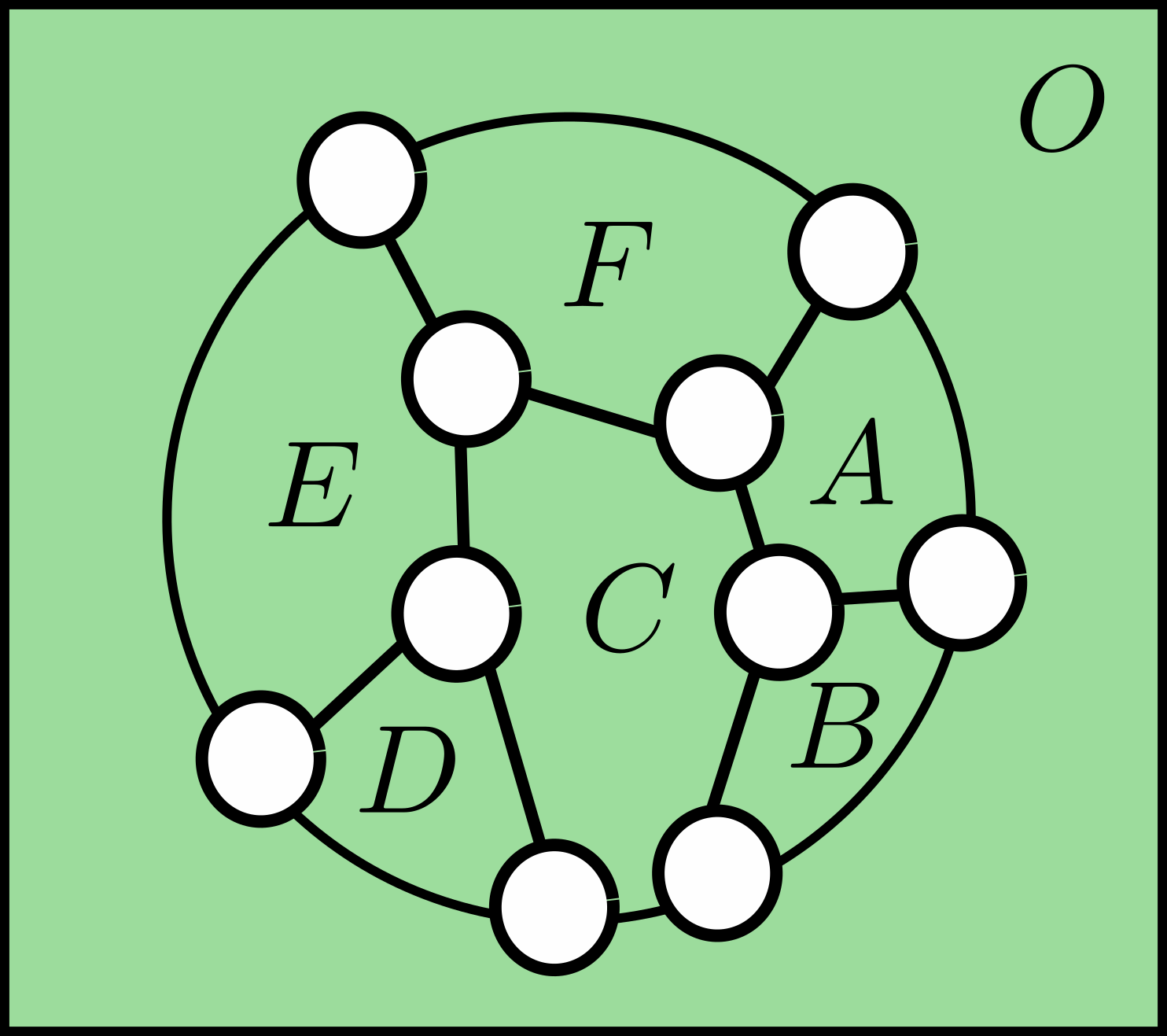}
    \caption{A disk with $6-$subregions obtained after smooth deformations which is used to obtain the value of the information quantity $\mathcal{Q}_{6,1}$, Eq.~\eqref{eq:q61} and $\mathcal{Q}_{6,2}$, Eq.~\eqref{eq:q62}. It turns out that the information quantity $\mathcal{Q}_{6,1}$ can not be used to probe the TEE however the latter quantity, $\mathcal{Q}_{6,2}$ can probe the TEE.}
    \label{fig:six_sub_regions}
\end{figure}
This is surprising as it naively seems to suggest that the space of allowed entropy vector space  of the unique ground state of a gapped TQFT is strictly contained within the HEC. However, the space of allowed entropy vector space in the case of degenerate ground states is expected to be larger than the HEC (and thus those states can violate the HEIs). Alternatively, it might be the case, that inclusion of the spurious entanglement entropy may violate these inequalities that seem to naively hold. We postpone these interesting questions for our future work and instead generalize the above observations below. We will need a few more definitions before we do that.

\begin{definition}[$k$-balanced HEIs]\label{def:superbalance}
A HEI $\mathcal{Q}$ is said to be $k$-balanced, iff every $k$-clustering of single subregions are balanced, i.e, they appear equal number of times with positive and negative coefficients.
\end{definition}
For example, consider the monogamy of mutual information (MMI), Eq.~\eqref{eq:Stopo_def3} involving three subregions. Each subregion $A, B, C$ are balanced, implying $1$-balanced, every $2$-clusters $AB,AC, BC$ are also balanced, implying $2$-balanced. However, the $3$-cluster $ABC$ is not balanced and so on. Note that $1$-balance and $2$-balance can be alternatively referred as balance and superbalance, respectively. We will recall the following proposition~\ref{con-2bal-satisfy} below which is conjectured on the basis of empirical observations and proved in appendix \ref{app:anyonic-topo-sb}.

\emph{Any $2$-balanced information quantity $\mathcal{Q}$ is topological with indefinite sign.}

Now we will state our general proposition applicable to all facet and non-facet inequalities.

\begin{proposition}\label{thm-general-satisfy}
    Any information quantity $\mathcal{Q}_i$ whose sum of coefficients is negative (positive) and satisfies $2$-balance, is both topological and sign-definite, on a 2d disk-like geometry.
\end{proposition}
\begin{proof}
By proposition \ref{con-2bal-satisfy}, $\mathcal{Q}_i$ is topological. If the sum of coefficients is negative (positive), then, each term contributes a $-\gamma$ as a result of connected geometry and therefore $\mathcal{Q}_i$ is sign-definite.
\end{proof}
This explains why facet HEIs $(n\geq 3)$ are both valid and topological. We will make this more concrete below.

\begin{lemma}\label{lem:balanced-hei}
    All facet HEIs, except subadditivity (SA) are $2$-balanced.
\end{lemma}
The proof is given in Ref.~\cite{He:2020superbalance}.

Since, all facet HEIs (except SA) satisfy lemma \ref{lem:balanced-hei}, \ref{lem:coeff-hei} and proposition \ref{thm-general-satisfy}, therefore, all facet HEIs are both valid and topological.

One may question the validity of the proposition \ref{con-2bal-satisfy}. In the first version of our work, we have empirically, we have checked all known (over 2000) holographic entropy inequalities $(n\geq 3)$ in the 2d disk geometry and found them to be both valid and topological. We have also tested them over many superbalanced information quantities that are holographically false to gather empirical evidence in favour of proposition \ref{thm-general-satisfy} and have not found any counterexample. In the revised version, we added a proof in appendix \ref{app:anyonic-topo-sb}.

Nevertheless, we will try to give physics motivation behind \ref{con-2bal-satisfy} from the perspective of entanglement entropy in quantum field theories~\cite{Casini:2004bw,Casini:2022rlv,Hubeny:2018ijt}. In a local QFT, the entanglement entropy of any finite region diverges. Balanced information quantities cancel these divergences, however, it treats the purifying region $O$ as a separate entity. Superbalanced quantities on the other hand treat the purifier $O$ at the same footing as other regions. This give rise to a scheme-independent cancellation of divergences. It has been argued that superbalanced quantities characterize certain topological properties of the configurations~\cite{Hubeny:2018ijt}. We believe that this picture is captured clearly in the TQFT setting where the total quantum dimension is related to the TEE. In the following section, we will suggest how our results are useful in the context of entanglement-based probes of topological phases.

\section{Discussions}
\label{sec:discussion}
We have shown that the holographic entropy inequalities are satisfied by the quantum entanglement entropy of the non-degenerate ground state of a topologically ordered two-dimensional medium with a mass gap. We also showed that these superbalanced inequalities can be used as probes for the topological entanglement entropy. In $(2+1)$d dimensions, the Bekenstein-Hawking entropy of BTZ has been proposed as a topological entanglement entropy \cite{McGough:2013gka}. Another recent work \cite{Bao:2024ixc} has also suggested that $AdS_3$ gravity can be described as a topological quantum field theory, and some instances of the $AdS_3/CFT_2$ correspondence may be viewed as a topological symmetry-preserving quantum RG flow. In the light of these developments, $3d$ gravity may serve as good inspiration for research on TQFTs and vice versa. We discuss some aspects of our results and future directions below.

\subsection{Entanglement-based probes of topological phases}
\label{subsec:tee-probe-tp}
Topological phase is associated with the existence of anyons. The TEE captures this anyonic content exactly because TEE is a function of the total quantum dimension of a given gapped Hamiltonian. Therefore calculating TEE of a system in a unique ground state is one of the possible pathways to detect topological order. Coming up with information quantities that can capture this long-range correlations, equivalently, the topological order is a non-trivial task. Our analysis sheds light on how one can write down some information quantities that can probe the TEE on $2d-$disk-like geometry. We showed that an information quantity which is superbalanced (see definition~\ref{def:superbalance}), can be used as a probe for the TEE. Hence, using the proposed prescription one can write down information quantities that can successfully probe the TEE and hence detect any topological order. Thus, a general formula for calculating the TEE is given by a superbalanced information quantity $\mathcal{Q}$ can be dubbed as,
\begin{equation}\label{eq:topo-probe-general-disk}
    \mathcal{Q}= -c \gamma = -c \log{\mathcal{D}},
\end{equation}
where $c$ is the sum of coefficients in $\mathcal{Q}$. If one is interested in obtaining $S_{\text{topo}}$, then 
\begin{equation}\label{eq:topo-probe-general-disk2}
    S_{\text{topo}}=\frac{-1}{c}\mathcal{Q}=\gamma,
\end{equation}
where the formula Eq.~\eqref{eq:topo-probe-general-disk} (and Eq.~\eqref{eq:topo-probe-general-disk2}) is applicable for the $2d$ disk geometry. By changing the topology, especially, connected components of the system, the formula needs to be appropriately modified taking the topology into account.

As shown in the appendix~\ref{sec:review-tee} and appendix~\ref{sec:cycli_inequalities}, the multi-information $\mathcal{I}_n$ and cyclic entropy inequalities $\mathcal{Q}_{2n+1}$ are good candidates for probing TEE . While the cyclic inequalities are motivated from holography, the multi-information doesn't have a supporting holographic counterpart. Both of these quantities are generalization of the MMI used in \cite{preskill-kitaev-tee}, which is given by $\mathcal{Q}_3$ and $\mathcal{I}_3$, respectively. The advantage of generalizing these quantities is that if the physical system in consideration is naturally divided into more than three regions, these quantities give a direct approach to calculate the TEE. Another advantage of using other facet HEI probes is that, it might be the case that for a non-trivial topology, one has the pre-factor vanishing for some inequalities having the same sum of coefficients, say $c$, in those cases, one has to choose a different probe, having a sum of coefficients, $c'$ (and $c'\neq c$).

To reiterate our classification of TEE probes, we have defined all superbalanced information quantities as fixed-topology TEE probes while all balanced information quantities are termed as fixed-geometry TEE probes, for which the 2d disk geometry is not suitable.

\subsection{Lessons from holography}
We will take inspiration from holography to comment on multi-information $\mathcal{I}_n$ being the potential signal of multipartite entanglement in gapped Hamiltonians. Recent works in holography \cite{Balasubramanian:2024ysu,Gadde:2024taa} propose various information quantities as measures (and signals) of multipartite entanglement entropy, each with their pros and cons. In particular, holographically, $\mathcal{I}_n$ is zero on separable quantum states and sign-indefinite for entangled states. Unfortunately, it is zero, for some states with true multiparite entanglement. Therefore, it is a good but not the best signal for probing multipartite entanglement in holography. In gapped Hamiltonians with unique ground states (i.e. topologies with 2d-disk), $\mathcal{I}_n=-\gamma$ (for $n\geq 3$), suggestive of $n$-party global correlations in the ground state. Of course, if one picks a geometry in which all subregions are disconnected, then $\mathcal{I}_n$ vanishes identically. So, we already know that $\mathcal{I}_n$ is a good signal. It would be interesting to import the alternative definitions of signals of multipartite entanglement from holography and probe them in topological systems. We will report results in this direction in future work.

\subsection{Including spurious entanglement entropies}
It was shown previously that, by including spurious contributions to entanglement entropy, the equality $\gamma=\log{\mathcal{D}}$ ceases to hold~\cite{PhysRevB.94.075151,PhysRevB.98.075131,PhysRevLett.122.140506,PhysRevB.100.115112,PhysRevResearch.2.032005,PhysRevLett.131.166601}. One limitation of our proof method is that we have first proved that the information quantities formed using HEIs are topological, i.e., they satisfy the equality $\gamma=\log{\mathcal{D}}$, and then use it to show that the HEIs are valid. One may argue that if instead one has an inequality $\gamma \geq \log{\mathcal{D}}$, are the facet HEIs still valid? Since the sum of coefficients for a facet HEI is negative, if there are physical processes that produce a positive contribution to the entanglement entropy (such as spurious contributions) for every or some subregions, depending on the relative magnitude of those contributions, it is a priori not clear if these inequalities continue to hold and one might expect to see violations similar to the case of quantum extremal surfaces in holography \cite{Akers:2021lms}. Perhaps, these exotic systems will shed more light on the nature of the HEC, if they cease to hold true for superbalanced information quantities with negative sum of coefficients that are not valid for holographic states, but continue to hold for facet HEIs. This is a very interesting question, and we will take a deep dive on this question in a future work. The answer will also shed light on the space of quantum states of gapped Hamiltonians with unique ground states.

\subsection{On changing the geometry}
Another limitation is that we have used the geometry of a 2d-disk for our proof. Does the proof hold when one considers other topologies? It is clear that the geometry where each region is disconnected from each other, all superbalanced quantities will identically evaluate to zero. However, we would like to believe that the facet HEIs would always give a non-negative value of the entanglement entropy for the subtraction scheme. It would be interesting if there is a geometry that contradicts this. For the case of connected geometries, where some HEIs may yield zero, others may not. In summary, we would like to conjecture that the full holographic entropy cone, inclusive of all facet HEIs is a probe of global correlations in gapped Hamiltonians. We also think that one may be able to come up with geometries that violate some or all of the superbalanced information quantities with negative sum of coefficients that do not hold holographically.

\subsection{Beyond gapped Hamiltonians}
Lastly, we would like to learn beyond holography and gapped Hamiltonians, and ask which inequalities hold true, in general, for all quantum states. All quantum entanglement entropy inequalities (QEEI) must be obeyed by the GHZ states. This puts a lower bound on the sum of the coefficients of the QEEIs~\cite{Bao:2024all-hei} to be zero. On the other hand, the 2d-disk geometry considered in our analysis puts the upper bound on the sum of coefficients to zero for the validity of such inequalities. Therefore, one only needs to look at those inequalities whose sum of coefficients are identically zero. We believe that these inequalities form the boundary that separates inequalities that are valid for topological systems from those that do not. Although one can easily find geometries that violate superbalanced inequalities that do not hold in holography whose sum of coefficients is zero~\footnote{We thank Michael Levin for working out this example with us.}, we leave it as a future exercise, to ascertain if the same is true for inequalities that hold holographically. This may shed light on whether the space of ground states of gapped Hamiltonians is closed or open.

\section{Acknowledgements}
We thank Ning Bao for suggesting this project direction. J.N. would like to thank Ning Bao, Keiichiro Furuya, Eugene Tang and particularly Michael Levin for useful discussions. S.S.S. would like thank Colleen Delaney, Bowen Yan, Ayush Raj and Soham Ray for useful discussions. We thank Ning Bao and Eugene Tang for comments on the draft. J.N. is supported by the NSF under Cooperative Agreement PHY2019786 and the Graduate Assistantship by the Department of Physics, Northeastern University. S.S.S. is partially supported by Graduate Assistantship by the Department of Physics and Astronomy, Purdue University.


\begin{appendix}

\begin{widetext}

\renewcommand{\thesection}{\Alph{section}}
\renewcommand{\theequation}{\thesection\arabic{equation}}
\renewcommand{\thefigure}{\thesection\arabic{figure}}
\renewcommand{\thetable}{\thesection\arabic{table}}

\makeatletter
\@addtoreset{equation}{section}
\@addtoreset{figure}{section}
\@addtoreset{table}{section}
\makeatother

\section{Multi-information as a probe for topological entanglement entropy}\label{sec:review-tee}
In this appendix, we are going to review the calculation for the topological entanglement entropy along the lines of Ref.~\cite{preskill-kitaev-tee}. Here, we will show that the $S_{\text{topo}}$ can be probed by $\mathcal{I}_{5}$ in a disk-like geometry with $5$ subregions, Fig.~\ref{fig:I5}b. The expression for $\mathcal{I}_{5}$ is given as follows,
\begin{widetext}
    \begin{align}
    S_{\text{topo}} \ (\text{$5$ subregions})\equiv -\mathcal{I}_{5} =& - S_{A} - S_{B} - S_{C} - S_{D} - S_{E} \\
    &+ S_{AB} + S_{AC}+ S_{AD}+ S_{AE}+ S_{BC}+ S_{BD}+ S_{BE}+ S_{CD}+ S_{CE}+ S_{DE}\\
    &- S_{ABC}- S_{ABD}- S_{ABE}- S_{ACD}- S_{ACE}- S_{ADE}- S_{BCD}\\
    &- S_{BCE}- S_{BDE}- S_{CDE}\\
    &+ S_{ABCD}+ S_{ABCE}+ S_{ACDE}+ S_{ABDE}+ S_{BCDE}\\
    &-S_{ABCDE}\,.
\end{align}
\end{widetext}
It is straightforward to show that in this case, all the boundary terms cancels exactly and under small and smooth deformations~\cite{preskill-kitaev-tee} $S_{\text{topo}}\equiv \mathcal{I}_{5}$ probes the topological entanglement entropy, Eq.~\eqref{eq:Stopo_def1}. The procedure to calculate $\mathcal{I}_{5}$ involves gluing together a time-reversal copy of the system at spatial infinity and making punctures at the intersection of the subregions. The resultant geometry is shown in the Fig.~\ref{fig:topological_pants}a.
\begin{figure*}
    \centering
    \includegraphics[width=\linewidth]{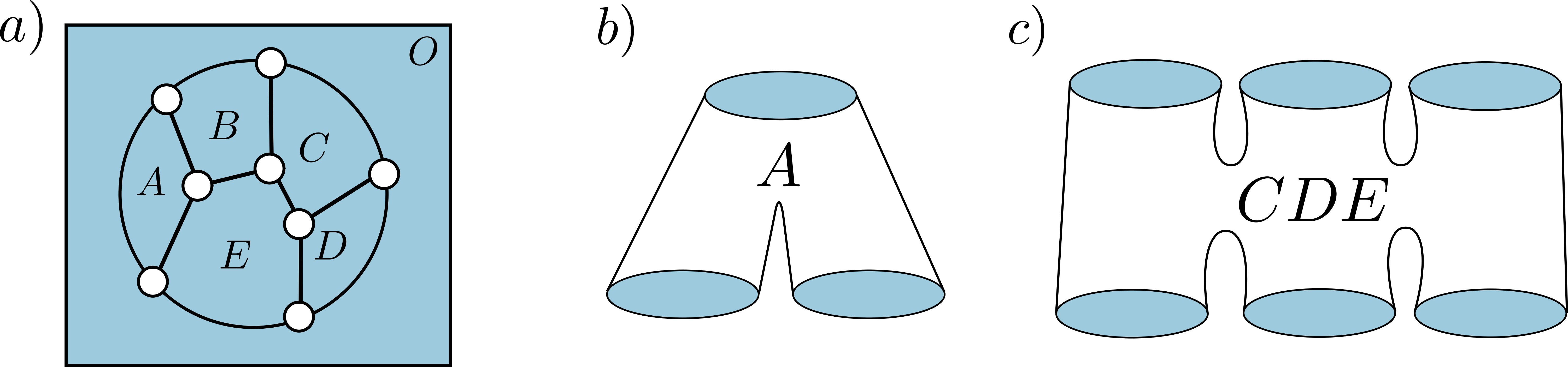}
    \caption{Geometry used to calculate the topological entanglement entropy ($S_{\text{topo}}$). (a) Stitching a time reversal copy at spatial infinity and making punctures at the intersections of subregions. (b) subregion $A$ is now identified with a sphere with $3-$punctures whose Von-neumann entropy is calculated using TQFT and is equal to $S_{A}=\mathcal{S}_{3} = 4\log\mathcal{D}-3\mathcal{K}$. (c)subregion $CDE$ is now identified with sphere with $6$ punctures and the Von-Neumann entropy is given as $S_{CDE}=\mathcal{S}_{6} = 10\log\mathcal{D} - 6\mathcal{K}$. Note that in this case only punctures at boundaries are contributing~\cite{preskill-kitaev-tee}}
    \label{fig:topological_pants}
\end{figure*}
If we look at any subregion, then we find that it is a sphere with $n$ number of punctures, e.g., region $A$ is identified with sphere with $3-$punctures and the region $CDE$ is identified with a sphere with $6-$punctures. The Von-Neumann entropy of any subregion can be computed using TQFT and is given by Eq.~\eqref{eq:sn_entropy}. This implies, we have,
\begin{widetext}
    \begin{align}
    S_{\text{topo}} \ (\text{$5$ subregions}) \equiv -\mathcal{I}_{5} =& \frac{1}{2}\big[-\mathcal{S}_{3}-\mathcal{S}_{4}-\mathcal{S}_{4}-\mathcal{S}_{3}-\mathcal{S}_{5} \\
    &+ \mathcal{S}_{5} + \mathcal{S}_{7} + \mathcal{S}_{6} + \mathcal{S}_{6} + \mathcal{S}_{6} + \mathcal{S}_{7} + \mathcal{S}_{7} + \mathcal{S}_{5} + \mathcal{S}_{7} + \mathcal{S}_{6} \\
    &- \mathcal{S}_{7} - \mathcal{S}_{8} - \mathcal{S}_{6} - \mathcal{S}_{8} - \mathcal{S}_{8} - \mathcal{S}_{7} - \mathcal{S}_{7} -\mathcal{S}_{7} - \mathcal{S}_{8} - \mathcal{S}_{6}\\
    &+ \mathcal{S}_{8} + \mathcal{S}_{6} + \mathcal{S}_{7} + \mathcal{S}_{7} + \mathcal{S}_{6} - \mathcal{S}_{5}\big] = \log \mathcal{D}\,,
\end{align}
\end{widetext}
by using Eq.~\eqref{eq:sn_entropy}. This shows that $\mathcal{I}_{5}$ also probes the topological entanglement entropy of a disk with $5$ subregions. Similar set of calculations can be performed on a disk with $3$ and $4$ subregions,
\begin{widetext}
    \begin{align}
    S_{\text{topo}} \ (\text{$3$ subregions})  \equiv -\mathcal{I}_{3} =& - S_{A} - S_{B} - S_{C} + S_{AB} + S_{BC} + S_{CA} - S_{ABC} =\frac{1}{2} \big( -4\mathcal{S}_{3}+3\mathcal{S}_{4} \big)\,,\\
    =& \frac{1}{2}\big( -16 \log \mathcal{D} + 12 \mathcal{K} + 18 \log \mathcal{D} - 12 \mathcal{K} \big) = \log \mathcal{D}\,, \\
    S_{\text{topo}} \ (\text{$4$ subregions})  \equiv -\mathcal{I}_{4} =& - S_{A} - S_{B} - S_{C} - S_{D} + S_{AB} + S_{AC} + S_{AD} + S_{BC} + S_{BD} + S_{CD} \\
    &- S_{ABC} - S_{ABD} - S_{BCD} - S_{ACD} + S_{ABCD}\,, \\
    =& 
    \frac{1}{2}\big(-2\mathcal{S}_{3}-\mathcal{S}_{4}+2\mathcal{S}_{5}\big)=\log \mathcal{D}\,.
\end{align}
\end{widetext}
Note that all the quantities $\mathcal{I}_{3}$, $\mathcal{I}_{4}$, and $\mathcal{I}_{5}$ give the same value for TEE, i.e. $S_{\text{topo}}=\log\mathcal{D}$. It is expected because the underlying topology does not change if we increase the number of subregions, therefore, we obtained the same TEE. We verified it numerically for up to $n=20$ subregions. Next, we shall consider the case where we calculate $\mathcal{I}_{3}$ for three arbitrarily chosen subregions out of $5$ available subregions on a disk, Fig.~\ref{fig:topological_pants}a. Consider a system consisting of $A$, $C$ and $E$, then we can probe the topological entanglement entropy as follows,
\begin{widetext}
    \begin{align}
    S_{\text{topo}} \ (\text{$5$ subregions}) = -\mathcal{I}_{3,5} \ (\text{$5$ subregions}) &=-\frac{1}{2} \big[S_{A} + S_{C} + S_{E} - S_{AC} - S_{AE} - S_{EC} + S_{AEC}\big] \,,\\
    &= -\frac{1}{2}\big[\mathcal{S}_{3} + \mathcal{S}_{4} + \mathcal{S}_{5} - \mathcal{S}_{7} - \mathcal{S}_{6} - \mathcal{S}_{7} 
    + \mathcal{S}_{8}\big]\,, \\
    &= \log \mathcal{D}\,.
\end{align}
\end{widetext}
Now, let us consider a more general case where we have $(2n+1)$-subregions (odd number of subregions) on a disk-like geometry, Fig.~\ref{fig:odd_number_sub_regions}a,c. Now, we shall probe the topological entanglement entropy by choosing $3$ subregions that are non-contiguous, say $A_{2n+1}$, $A_{2}$ and $A_{4}$. In this case, we can write the following,
\begin{widetext}
    \begin{align}
    S_{\text{topo}} \ (\text{$2n+1$ subregions}) = -\mathcal{I}_{3,2n+1} &=S_{A_{2n+1}} + S_{A_{2}} + S_{A_{4}} - S_{A_{2n+1}A_{2}} - S_{A_{2n+1}A_{4}} - S_{A_{2}A_{4}} + S_{A_{2n+1}A_{2}A_{4}} \,,\\
    &= -\frac{1}{2}\big[\mathcal{S}_{2n+1} + \mathcal{S}_{4} + \mathcal{S}_{4} - \mathcal{S}_{2n+3} - \mathcal{S}_{2n+3} - \mathcal{S}_{8} 
    + \mathcal{S}_{2n+5}\big]\,, \\
    &= \log \mathcal{D}\,.
\end{align}
\end{widetext}
where we have used the general formula for the VN entropy of a sphere with $n-$number of punctures, Eq.~\eqref{eq:sn_entropy}. In all of the examples we demonstrated that the there are many ways to probe the topological entanglement entropy and all of them yields $S_{\text{topo}}=\log \mathcal{D}$. Since the underlying geometry is the same, therefore the TEE computed is also same.
\section{Proof for the general case of cyclic inequalities}\label{sec:cycli_inequalities}
In this section we are going to look at the general case of cyclic inequalities for $(2n+1)$-subregions (odd number of subregions) on a disk-like geometry, Fig.~\ref{fig:odd_number_sub_regions}. Expression for the cyclic inequality can be written as follows,
\begin{equation}
    \label{eq:cyclic-ineqs}
    \sum_{i=1}^{2n+1} S_{a_i^+} \geq \sum_{i=1}^{2n+1} S_{a_i^-}+ S_{A}\,,
\end{equation}
where we have defined $A=\{a_1,\dots, a_{2n+1}\}$ to be the set of $2n+1$ subregions, $S_{A}=S_{a_1\dots a_{2n+1}}$ and
\begin{equation}\label{eq:a_i^k}
    a^{(k)}_i = a_i \cdots a_{i+k-1},\;a^{\pm}_i := a_i^{(\frac{2n+1\pm 1}{2})} = a_i^{n\pm 1}\,.
\end{equation}
The corresponding information quantity turns out to be able to probe the topological entanglement entropy,
\begin{equation}
    S_{\text{topo}}=\sum_{i=1}^{2n+1} S_{a_i^+}-\sum_{i=1}^{2n+1} S_{a_i^-}-S_{A}.
\end{equation}
This was proven in Ref.~\cite{Bao:2015gapped} by assuming the area law Eq.~\eqref{eq:Stopo_def1}. We quote their result on the derived formula for $\gamma$ below,
\begin{widetext}
    \begin{align}\label{eq:ning-gamma}
    S_{\text{topo}}= -\gamma\left\{\sum_{i=1}^{2n+1}\left(b_0 [\partial(A_i\dots A_{i+n}]-b_0 [\partial(A_{i+1}\dots A_{i+n}]\right)+b_0 [\partial(A_1\dots A_{2n+1})] \right\},
\end{align}
\end{widetext}
where $b_0[\partial A]$ denotes the zeroth Betti number (the number of connected components) of the boundary of a region A, and all indices are taken modulo $(2n + 1)$. Now we will give a TQFT proof of the same for a fixed topology, i.e., the 2d disk. The first step involves a continuous deformation of the geometry where we have intersection of at most $3$ subregions. Second step involves gluing together a time-reversal copy of the system and making punctures at each intersection point, as shown in the Fig.~\ref{fig:odd_number_sub_regions}. 
\begin{figure*}
    \centering
    \includegraphics[width=1.0\linewidth]{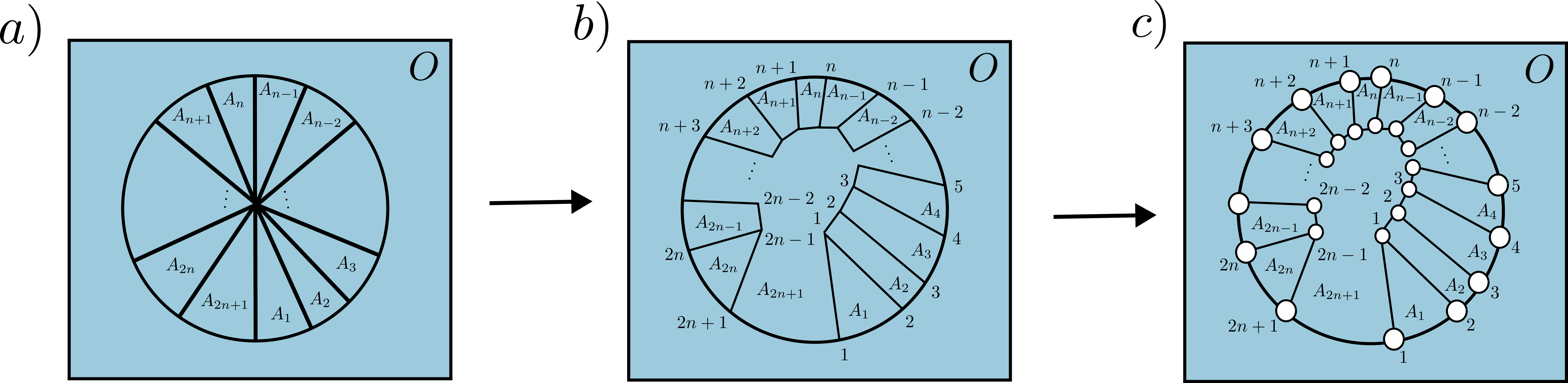}
    \caption{(a) Extracting the topological entanglement entropy of disk geometry with with $2n+1$ subregions. (b) First step involves deforming the geometry where we have at most intersection of $3$ subregions. (c) Gluing a time-reversal copy of the system and making punctures at the intersection of $3$ subregions. As a result of this Von-Neumann entropy of a subregion is equal to Von-Neumann entropy of sphere with punctures where each puncture host a anyon of a particular type.}
    \label{fig:odd_number_sub_regions}
\end{figure*}
The final step involves identifying the Von Neumann entropy of the subregions in terms of a sphere with particular number of punctures. This is done as follows,
\begin{widetext}
    \begin{align}
    \sum_{i=1}^{2n+1}S_{a_{i}}^{+} & =S_{a_{1}}^{+}+\sum_{i=2}^{n-1}S_{a_{i}}^{+}+S_{a_{n}}^{+}+S_{a_{n+1}}^{+}+S_{a_{n+2}}^{+}+\sum_{i=n+3}^{2n-1}S_{a_{i}}^{+}+S_{2n}^{+}+S_{a_{2n+1}}^{+}\,,\\
    & =S_{a_{1}}^{+}+(n-1-2+1)S_{a_{i}}^{+}+S_{a_{n}}^{+}+S_{a_{n+1}}^{+}+S_{a_{n+2}}^{+}+(2n-1-n-3+1)S_{a_{n+3}}^{+}++S_{2n}^{+}+S_{a_{2n+1}}^{+}\,,\\
 & =S_{a_{1}}^{+}+(n-2)S_{a_{2}}^{+}+S_{a_{n}}^{+}+S_{a_{n+1}}^{+}+S_{a_{n+2}}^{+}+(n-3)S_{a_{n+3}}^{+}+S_{2n}^{+}+S_{a_{2n+1}}^{+}\,,\\
 & =\frac{1}{2}\big[\mathcal{S}_{2n+3}+(n-2)\mathcal{S}_{2n+4}+\mathcal{S}_{2n+3}+\mathcal{S}_{2n+2}+\mathcal{S}_{2n+3}+(n-3)\mathcal{S}_{2n+3}+\mathcal{S}_{2n+3}+\mathcal{S}_{2n+2}\big]\,,\\
  & =\frac{1}{2}\big[2\mathcal{S}_{2n+2}+(n+1)\mathcal{S}_{2n+3}+(n-2)\mathcal{S}_{2n+4}\big]\,,\\
  \sum_{i=1}^{2n+1}S_{a_{i}}^{-} & =S_{a_{1}}^{-}+\sum_{i=2}^{n}S_{a_{i}}^{-}+S_{a_{n+1}}^{-}+S_{a_{n+2}}^{-}+S_{a_{n+3}}^{-}+\sum_{i=n+4}^{2n-1}S_{a_{i}}^{-}+S_{2n}^{-}+S_{a_{2n+1}}^{-}\,,\\
 & =\frac{1}{2}\big[\mathcal{S}_{2n+1}+(n-1)\mathcal{S}_{2n+2}+\mathcal{S}_{2n+1}+\mathcal{S}_{2n+2}+\mathcal{S}_{2n+3}+(n-4)\mathcal{S}_{2n+3}+\mathcal{S}_{2n+3}+\mathcal{S}_{2n+2}\big]\,,\\
 & =\frac{1}{2}\big[2\mathcal{S}_{2n+1}+(n+1)\mathcal{S}_{2n+2}+(n-2)\mathcal{S}_{2n+3}\big]\,.
\end{align}
\end{widetext}
The factor of half comes from the fact that we have doubled the system by gluing it with its time-reversal copy. Similar set of arguments and calculations gives,
the last term $S_{A}=\frac{1}{2}\mathcal{S}_{2n+1}$ and this implies, we have,
\begin{align}
    \sum_{i=1}^{2n+1}S_{a_{i}}^{+} - \sum_{i=1}^{2n+1}S_{a_{i}}^{-} - S_{A}  
    =& \log \mathcal{D} \,.
\end{align}
where we have plugged in Eq.~\eqref{eq:sn_entropy} to obtain the final result.

\section{Proof of holographic entropy inequalities in non-degenerate ground states of TQFT}\label{sec:HEI-proof}

Now we will introduce the tripartite form of the facet HEIs given in \cite{Hernandez-Cuenca:2023iqh}.
\begin{definition}[Tripartite form \cite{Hernandez-Cuenca:2023iqh}]\label{def:tripartite}
    An information quantity $\mathcal{Q}$ is said to be in the tripartite form if it is expressed as
    \begin{equation}\label{eq:tripartite}
        \mathcal{Q}=\sum_i -I_3(X_i:Y_i:Z_i | W_i)
    \end{equation}
    where the arguments $X_i, Y_i, Z_i, Wi \subset [N]$ are disjoint subsystems, the sum runs over any finite number of terms, and we allow for the conditioning to trivialize, $W_i = \emptyset$, in which case $I_3(X_i:Y_i:Z_i | \emptyset)= I_3(X_i:Y_i:Z_i)$ and, they are defined to be
    \begin{equation}
        I_3(X_i:Y_i:Z_i | W_i)=I_3(X_i:Y_i:Z_i W_i)-I_3(X_i:Y_i:W_i)
    \end{equation}
    and,
    \begin{equation}
        I_3(X_i:Y_i:Z_i)= X_i + Y_i + Z_i - X_i Y_i - X_i Z_i - Y_i Z_i + X_i Y_i Z_i
    \end{equation}
    We denote $I^pC^q$ for a $\mathcal{Q}$ that has $p$ number of $-I_3(X_i: Y_i: Z_i)$ and $q$ number of $-I_3(X_i:Y_i:Z_i | W_i)$ terms in the sum (\ref{eq:tripartite}).
\end{definition}

We borrow the following conjecture from \cite{Hernandez-Cuenca:2023iqh}, which holds true for all known superbalanced facet HEIs.

\begin{proposition}\label{conj:sergio}
     All facet inequalities (except SA) are expressible in the  $I^pC^q$ form with $p\geq 1$ and $q\geq 0$.
\end{proposition}
The definition \ref{def:tripartite}, taken together with the proposition \ref{conj:sergio}, gives a straightforward proof for lemma \ref{lem:coeff-hei} used in the text.

\begin{theorem}\label{thm:hei-satisfy-alter}
All facet HEIs of the HEC $(n\geq 3)$ are satisfied by the quantum entanglement entropy in the unique ground state of a topologically ordered two-dimensional medium with a mass gap. 
\end{theorem}
\begin{proof}
From proposition \ref{thm:hei-topological}, we have all facet HEIs are topological. We will now prove that they are also non-negative. 
Following section~\ref{sec:tqft-tech} and appendix~\ref{sec:review-tee}, we have $\gamma=\log{\mathcal{D}}$ on the $2d$ disk geometry. 

The conditional multi-information terms (denoted as $C$) in a facet HEI expressed in triparite form (Eq. \ref{eq:tripartite}) is a difference of two multi-information, both of which are topological and therefore, cancels each other. The only contribution comes from the multi-information terms (denoted as $I$). Since, $I$ is topological and $I\leq 0$ (by extension of proposition \ref{thm:In-topological}), and by proposition \ref{conj:sergio}, $p\geq1$, thus $\mathcal{Q}\geq 0$. Therefore, the facet HEIs are non-negative on a $2d$ disk geometry.
\end{proof}

\section{A Heuristic Proof of Superbalanced $\implies$ Topological}
In this appendix, we give a heuristic proof based on the techniques introduced in \cite{preskill-kitaev-tee} that superbalancedness information quantities are topological. Consider decomposing a information quantity $\mathcal{Q}$ into a null reduction $\mathcal{Q}_{\downarrow A_n}$ on some party $A_n$ such that \cite{Grimaldi:2025jad}
\begin{equation}\label{eq:null-reduce-on-An}
    \mathcal{Q}= \mathcal{Q}_{\downarrow A_n} + \overline{\mathcal{Q}_{\downarrow A_n}},
\end{equation}
where $\mathcal{Q}_{\downarrow A_n}$ contains all the terms including the party $A_n$, and $\overline{\mathcal{Q}_{\downarrow A_n}}$ contains all the terms excluding the party $A_n$. It was shown in \cite{Grimaldi:2025jad} that both $\mathcal{Q}_{\downarrow A_n}$ and $\overline{\mathcal{Q}_{\downarrow A_n}}$ are balanced when $\mathcal{Q}$ is superbalanced.

For example, if the information quantity $\mathcal{Q}$ is the MMI,
\begin{equation}
    \mathcal{Q}^{(MMI)}= S_{AB} + S_{AC} + S_{BC} - S_A - S_B - S_C - S_{ABC},
\end{equation}
then, the null reduction on $A$ is given by,
\begin{equation}\label{eq:mmi-null-A}
    \mathcal{Q}^{(MMI)}_{\downarrow A}= S_{AB} + S_{AC} - S_A - S_{ABC},
\end{equation}
and the remaining part of the information quantity is given by,
\begin{equation}
    \overline{\mathcal{Q}^{(MMI)}_{\downarrow A}}= S_{BC} - S_B - S_C.
\end{equation}

\label{app:heuristic-proof}
\begin{figure}
    \centering
    \includegraphics[width=0.75\linewidth]{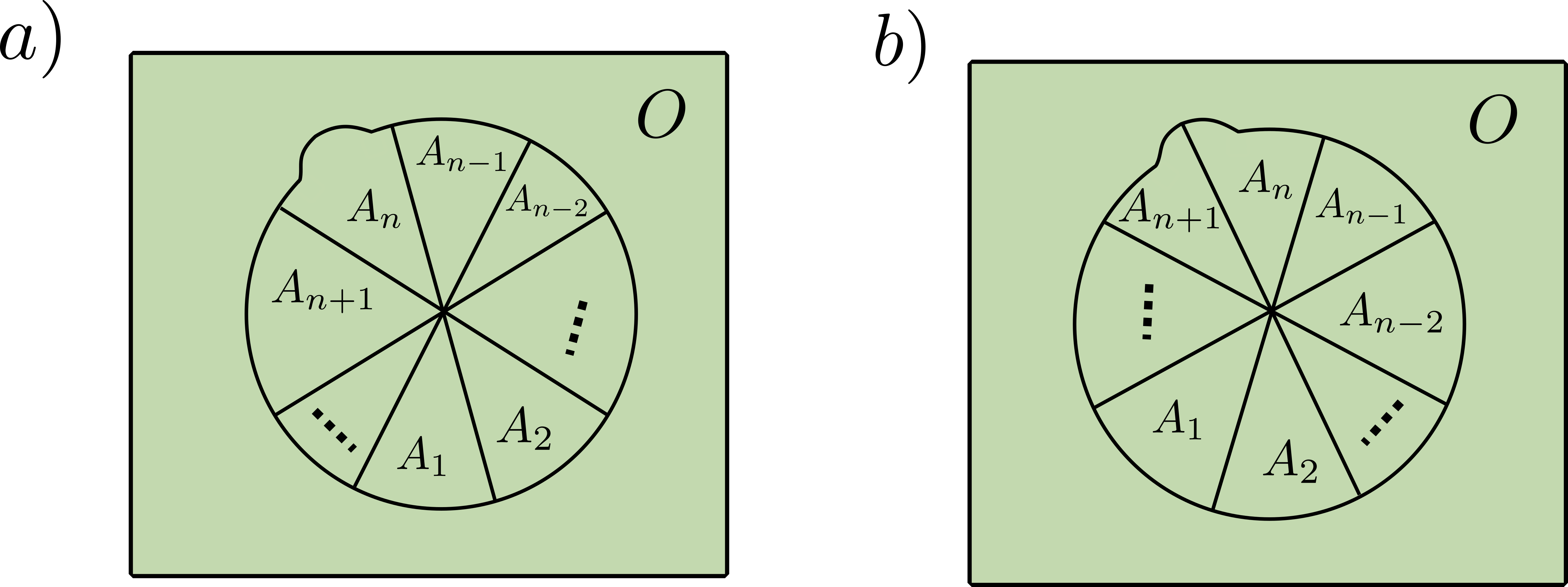}
    \caption{Two possible deformations of the region $A_n$ at (a) two-boundary junction, and (b) three-boundary junction.}
    \label{fig:placeholder}
\end{figure}

Coming back to generality, let us consider that there is a deformation in the region $A_n$. There are two possibilities that we are going to consider. First, the deformation is at a two-boundary junction. In this case, all terms containing $S_{A_n}$ gets replaced with $S_{A_n}+\Delta S_{A_n}$. We can write,
\begin{equation}\label{eq:null-reduce-on-An-delta}
    \Delta\mathcal{Q}= \Delta\mathcal{Q}_{\downarrow A_n} + \Delta\overline{\mathcal{Q}_{\downarrow A_n}}=\Delta\mathcal{Q}_{\downarrow A_n}=0,
\end{equation}
where $\Delta\overline{\mathcal{Q}_{\downarrow A_n}}=0$ because this collection does not contain any term involving $S_{A_n}$, and $\Delta\mathcal{Q}_{\downarrow A_n}=0$ due to balance. The second case is when the deformation is at a three-boundary junction, where one can argue similarly. However, one needs to be careful with the counting as the two decompositions do not vanish separately. We are not going to do that heuristic proof here and instead give an anyonic derivation below.

\section{Anyonic Derivation of Superbalanced $\implies$ Topological}\label{app:anyonic-topo-sb}

Consider the geometry we have i.e. a disk with $n-$slices of pizza similar to Fig.~\ref{fig:odd_number_sub_regions} but with $n$ number of slices. Next, we smoothly deform it to consist only vertices that are triple point. A triple point $t_{j}$ involves exactly three regions from the set $\{A_{1},\cdots,A_{n},O\}$ and we call them $\{R_{j}^{(1)},R_{j}^{(2)},R_{j}^{(3)}\}$. 

\begin{definition}
For a subset $I\subseteq\{A_{1},\cdots,A_{n},O\}$ and
triple point $t_{j}$, define the indicator function
\begin{equation}
\sigma_{j}(I):=\begin{cases}
1 & \text{if \ensuremath{t_{j}}}\text{ lies on the boundary of region \ensuremath{I}}\\
0 & \text{otherwise }
\end{cases}
\end{equation}
\end{definition}
Physically this captures the following on the disk with $O$ being
the exterior:
\begin{enumerate}
\item $\sigma_{j}(I)=0$ if all the three regions $\{R_{j}^{(1)},R_{j}^{(2)},R_{j}^{(3)}\}$ at $t_{j}$ are subregions
contained in $I$ i.e., $\{R_{j}^{(1)},R_{j}^{(2)},R_{j}^{(3)}\} \subseteq I$ or if all the three regions $\{R_{j}^{(1)},R_{j}^{(2)},R_{j}^{(3)}\}$ at $t_{j}$ are in
the complement i.e.,  $\{R_{j}^{(1)},R_{j}^{(2)},R_{j}^{(3)}\}\subseteq \{A_{1},\cdots,A_{n},O\}\setminus I$.
\item $\sigma_{j}(I)=1$ otherwise, referring to the case where the triple
$t_{j}$ lies at the boundary of $I$ and its complement.
\end{enumerate}
\begin{definition}
An information quantity $\mathcal{Q}=\sum_{I}a_{I}S(I)$
is called $2-$balanced or superbalanced for every single party $A_{k} \in \{A_{1},\cdots,A_{n},O\}$
we have,
\begin{equation}
\sum_{A_{k}\subseteq I}a_{I}=0\,,
\end{equation}
and for every pair $A_{k}$, $A_{l}$ with $k\neq l$, we have,
\begin{equation}
\sum_{\{A_{k},A_{l}\}\subseteq I}a_{I}=0\,.
\end{equation}
\end{definition}
For later convenience, we also define,
\begin{equation}
c:=\sum_{I}a_{I} \,.
\end{equation}
The Von-Neumann entropy for a $2d$ disk geometry is given as follows,
\begin{align}
\mathcal{Q} & =\sum_{I}a_{I}\Big[(p_{I}-1)\log\mathcal{D}-\frac{p_{I}}{2}\mathcal{K}\Big] \\ 
&=\log\mathcal{D}\cdot\sum_{I}a_{I}(p_{I}-1)-\frac{1}{2}\Big[\sum_{I}a_{I}p_{I}\Big]\cdot\mathcal{K}\,,
\end{align}
where $p_{I}$ refers to the number of triple points on the boundary of the region $I$. For our purpose we also assume that the region $I$ is made up of the elements from the set $\{A_{1},\cdots,A_{n}\}$. For proving that $\mathcal{Q}$ is topological entanglement entropy,
we must prove,
\begin{equation}
\sum_{I}a_{I}p_{I}=0\,.
\end{equation}
Note that we can also write,
\begin{equation}
\sum_{I}a_{I}p_{I}=\sum_{I}a_{I}\sum_{j=1}^{T}\sigma_{j}(I)=\sum_{j=1}^{T}\sum_{I}a_{I}\sigma_{j}(I)\,,
\end{equation}
where $T$ is the total number of triple points in the $2d$ disk geometry. Next, note that each of the triple point is either made up of regions
$\{O,A_{a},A_{b}\}$ or of the regions $\{A_{a},A_{b},A_{c}\}$. We shall address these two cases one by one.
\subsubsection*{Case I: Triple point $t_{j}$ is the intersection of regions of the type $\{O,A_{a},A_{b}\}$}
Now, the triple point $t_{j}$ lies on the boundary if either $A_{a}\subseteq I$
or $A_{b}\subseteq I$. This implies, we have,
\begin{equation}
\sigma_{j}(I)=\begin{cases}
0 & \text{if }A_{a}\not\subseteq I\text{ and }A_{b}\not\subseteq I\\
1 & \text{if }A_{a}\subseteq I\text{ or }A_{b}\subseteq I
\end{cases}
\end{equation}
Hence, we have,
\begin{equation}
\sum_{I}a_{I}\sigma_{j}(I)=\sum_{I}a_{I}-\sum_{I;A_{a}\not\subseteq I;A_{b}\not\subseteq I}a_{I} \,,
\end{equation}
where we have used the substraction scheme. Let us use set theory to
calculate the following,
\begin{align}
\sum_{I;A_{a}\not\subseteq I;A_{b}\not\subseteq I}a_{I} & =\sum_{I}a_{I}-\cancelto{0}{\sum_{A_{a}\subseteq I}a_{I}}-\cancelto{0}{\sum_{A_{b}\subseteq I}a_{I}}+\cancelto{0}{\sum_{\{A_{a},A_{b}\}\subseteq I}a_{I}}\,,\\
 & =\sum_{I}a_{I}\,,\\
 & =c\,,
\end{align}
where we have used the fact information quantity $\mathcal{Q}$ is
superbalanced. Hence, we have the following,
\begin{align}
\sum_{I}a_{I}\sigma_{j}(I) & =\sum_{I}a_{I}-\sum_{I;A_{a}\not\subseteq I;A_{b}\not\subseteq I}a_{I}\,,\\
 & =\sum_{I}a_{I}-\sum_{I}a_{I}\,,\\
\sum_{I}a_{I}\sigma_{j}(I) & =0\,.
\end{align}
This implies we have,
\begin{equation}
\sum_{I}a_{I}p_{I}=\sum_{I}a_{I}\sum_{j=1}^{T}\sigma_{j}(I)=\sum_{j=1}^{T}\sum_{I}a_{I}\sigma_{j}(I)=0\,.
\end{equation}

\subsubsection*{Case II: Triple point $t_{j}$ is the intersection of regions of
the type $\{A_{a},A_{b},A_{c}\}$}

In this case the triple point $t_{j}$ is on the boundary unless $\{A_{a},A_{b},A_{c}\}\subseteq I$
or $\{A_{a},A_{b},A_{c}\}\cap I=\emptyset$. Concretely we can write,
\begin{equation}
\sigma_{j}(I)=\begin{cases}
0 & \text{if }\{A_{a},A_{b},A_{c}\}\subseteq I\text{ or }\{A_{a},A_{b},A_{c}\}\cap I=\emptyset\\
1 & \text{otherwise}
\end{cases}
\end{equation}
Next, we can write the following,
\begin{equation}
\sum_{I}a_{I}\sigma_{j}(I)=\sum_{I}a_{I}-\sum_{\{A_{a},A_{b},A_{c}\}\subseteq I}a_{I}-\sum_{\{A_{a},A_{b},A_{c}\}\cap I=\emptyset}a_{I}
\end{equation}
note that there is no intersection part because they mutally exclusive.
Next, note the following,
\begin{align}
\sum_{\{A_{a},A_{b},A_{c}\}\cap I=\emptyset}a_{I} & =\sum_{I}a_{I}-\cancelto{0}{\sum_{A_{a}\subseteq I}a_{I}}-\cancelto{0}{\sum_{A_{b}\subseteq I}a_{I}}-\cancelto{0}{\sum_{A_{c}\subseteq I}a_{I}} \,,\\
 & \quad+\cancelto{0}{\sum_{\{A_{a},A_{b}\}\subseteq I}a_{I}}+\cancelto{0}{\sum_{\{A_{a},A_{c}\}\subseteq I}a_{I}}+\cancelto{0}{\sum_{\{A_{b},A_{c}\}\subseteq I}a_{I}}\,,\\
 & \quad-\sum_{\{A_{a},A_{b},A_{c}\}\subseteq I}a_{I}\,,\\
 & =\sum_{I}a_{I}-\sum_{\{A_{a},A_{b},A_{c}\}\subseteq I}a_{I}\,,
\end{align}
where we used the fact that information quantity $Q$ is superbalanced. This implies, we have,
\begin{align}
\sum_{I}a_{I}\sigma_{j}(I) & =\sum_{I}a_{I}-\sum_{\{A_{a},A_{b},A_{c}\}\subseteq I}a_{I}-\sum_{\{A_{a},A_{b},A_{c}\}\cap I=\emptyset}a_{I}\,,\\
 & =\cancel{\sum_{I}a_{I}}-\cancel{\sum_{\{A_{a},A_{b},A_{c}\}\subseteq I}a_{I}}-\cancel{\sum_{I}a_{I}}+\cancel{\sum_{\{A_{a},A_{b},A_{c}\}\subseteq I}a_{I}}\,,\\
 & =0\,.
\end{align}
Hence using case I and case II, we find the the quantity,
\begin{equation}
\sum_{I}a_{I}\sigma_{j}(I)=0\Rightarrow\sum_{j=1}^{T}\sum_{I}a_{I}\sigma_{j}(I)=\sum_{I}a_{I}p_{I}=0\,.
\end{equation}
and therefore the information superbalanced quantity $\mathcal{Q}=\sum_{I}a_{I}S(I)$
is topological.

\end{widetext}

\end{appendix}

\bibliography{ref}

\end{document}